\documentclass[a4paper,11pt]{article}

\usepackage[utf8]{inputenc}
\usepackage[english]{babel}
\usepackage[margin=1in]{geometry}
\usepackage{amsmath,amssymb,mathtools}
\usepackage{amsthm}
\usepackage{booktabs}
\usepackage{array}
\usepackage{capt-of}
\usepackage{microtype}
\usepackage{xcolor}
\usepackage{tikz}
\usetikzlibrary{arrows.meta,positioning,decorations.pathreplacing}
\usepackage{hyperref}
\hypersetup{hidelinks}

\newtheorem{theorem}{Theorem}

\newtheorem{lemma}[theorem]{Lemma}
\newtheorem{corollary}[theorem]{Corollary}
\theoremstyle{remark}

\newcommand{\U}{\mathrm{U}}

\newcommand{\E}{\mathbb{E}}
\newcommand{\Pp}{\mathbb{P}}
\newcommand{\C}{\mathbb{C}}
\newcommand{\R}{\mathbb{R}}
\newcommand{\tr}{\operatorname{tr}}
\newcommand{\diag}{\operatorname{diag}}
\newcommand{\Lip}{\operatorname{Lip}}
\newcommand{\Var}{\operatorname{Var}}
\newcommand{\HS}{\mathrm{HS}}
\newcolumntype{L}[1]{>{\raggedright\arraybackslash}p{#1}}

\begin{document}

\title{Weak Typicality of von Neumann Entanglement Entropy in Gaussian Boson Sampling}

\author{Hongru Zhao\\
\small School of Statistics, University of Minnesota, Minneapolis, Minnesota 55455, USA\\
\small \texttt{zhao1118@umn.edu}}
\date{}

\maketitle

\begin{abstract}
We study the von Neumann entanglement entropy generated by a Haar distributed
passive interferometer acting on \(n\) equally squeezed input modes with fixed
nonzero squeezing strength \(s\).  Previous work established proportional
weak typicality for integer R\'{e}nyi orders \(\alpha\ge2\) and stated a
sublinear von Neumann result, while the proportional von Neumann case remained
open.  For a subsystem of \(k_n\) modes satisfying
\(k_n/n\to r\in(0,1)\), we prove the summable relative concentration bound
\[
 \mathbb{P}\!\left(
 \left|\frac{S_{1,n}}{\mathbb{E}S_{1,n}}-1\right|
 \ge\varepsilon\right)
 \le
 2\exp\!\left[-\frac{c_{s,r}\varepsilon^2n^2}{\log^2(en)}\right].
\]
The proof represents the entropy as a singular value statistic of a principal
block of \(UU^{\mathsf T}\), where \(U\) denotes the unitary interferometer.
It regularizes the logarithmic
singularity at the pure symplectic endpoint and applies concentration on the
unitary group.  The result establishes
proportional von Neumann weak typicality and further implies almost sure
relative convergence, a typical volume law, and the variance bound
\(\operatorname{Var}S_{1,n}=O_s(\log^2 n)\).  An accompanying Lean~4
development verifies the proof chain.
\end{abstract}

\section{Introduction}

Gaussian boson sampling (GBS) prepares squeezed states, mixes the modes in a
linear optical interferometer, and measures output photon numbers
~\cite{Hamilton2017,Deshpande2022,HangleiterEisert2023}.  Before measurement,
the ideal output is a pure bosonic Gaussian state.  Its entanglement across a
division of the modes is measured by the von Neumann entropy of either reduced
state.  Gaussian states admit an exact description through their first and
second moments~\cite{Weedbrook2012,Serafini2017}, making this entropy accessible
to analysis.

The finite dimensional Page curve was introduced by Page, and Sen subsequently
proved Page's formula~\cite{Page1993,Sen1996}.  Bosonic Gaussian Page curves
for passive random interferometers were studied by Iosue \emph{et al.}
~\cite{Iosue2023}, and Youm \emph{et al.} derived the asymptotically exact von
Neumann mean curve for the fixed equal squeezing ensemble~\cite{Youm2025}.
Earlier rigorous results covered several related but different settings.
Serafini, Dahlsten, Plenio, and collaborators
~\cite{SerafiniDahlstenPlenio2007,SerafiniDahlstenGrossPlenio2007} established
fixed finite subsystem concentration for microcanonical hard energy and
canonical Boltzmann measures on pure Gaussian states.  Fukuda and
Koenig~\cite{FukudaKoenig2019} studied deterministic squeezed inputs under conditions
that exclude proportional subsystems.  Their exact assumptions are recorded
in Table~\ref{tab:prior-landscape}.

For the fixed equal squeezing Haar ensemble studied here, Iosue \emph{et al.}
~\cite{Iosue2023} proved proportional weak typicality for R\'{e}nyi entropy of order
2 and stated von Neumann weak typicality for sublinear subsystems.  Youm \emph{et al.}
~\cite{Youm2025} obtained asymptotically exact mean Page curves and extended
the proportional R\'{e}nyi result to every integer order \(\alpha\ge2\).  They
explicitly left proportional von Neumann typicality open.

The difficulty is that the derivative of the one mode entropy profile diverges
logarithmically as a reduced symplectic eigenvalue approaches its pure value
\(\nu=1\).  We avoid a direct variance calculation by expressing the entropy
through singular values, regularizing this endpoint, and applying Haar
concentration.  The resulting tail bound proves weak typicality and also gives
a typical volume law, almost sure convergence, and the variance bound
\(\Var S_{1,n}=O_s(\log^2 n)\).

\begin{figure}[t]
\centering
\begin{tikzpicture}[
  font=\scriptsize,
  box/.style={draw,rounded corners,minimum width=1.28cm,minimum height=0.72cm,align=center,inner sep=3pt},
  line/.style={-Latex,semithick}
]
\node[box,fill=gray!10] (u) {\(U\)};
\node[box,fill=blue!7,right=0.35cm of u] (q) {\(UU^{\mathsf T}\)};
\node[box,fill=blue!7,right=0.35cm of q] (b) {\(B_U\)};
\node[box,fill=orange!10,right=0.35cm of b] (tau) {\(\{\tau_j\}\)};
\node[box,fill=orange!10,below=0.38cm of tau] (nu) {\(\{\nu_j\}\)};
\node[box,fill=green!8,left=0.35cm of nu] (s) {\(S_{1,n}\)};
\node[box,fill=green!8,left=0.35cm of s] (f) {\(F_{n,\eta}\)};
\node[box,fill=purple!8,left=0.35cm of f] (c) {Haar\\tail};
\foreach \a/\bb in {u/q,q/b,b/tau,tau/nu,nu/s,s/f,f/c}{\draw[line] (\a) -- (\bb);}
\end{tikzpicture}
\caption{Proof pipeline, read across the upper row and then back across the lower row.  Here \(U\) is the Haar distributed interferometer, \(B_U\) is the selected principal block of \(UU^{\mathsf T}\), \(\tau_j\) are its singular values, \(\nu_j\) are the resulting symplectic eigenvalues, \(S_{1,n}\) is the von Neumann entropy, and \(F_{n,\eta}\) is its endpoint regularized approximation.  See Eqs.~\eqref{eq:B-definition}, \eqref{eq:tau-definition}, \eqref{eq:symplectic-eigenvalues}, \eqref{eq:entropy-spectral-statistic}, and \eqref{eq:regularized-statistic}.  Haar concentration is applied to \(F_{n,\eta}\).}
\label{fig:model}
\end{figure}
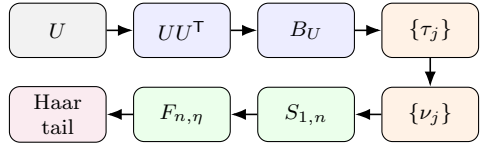

\newcommand{\priorlandscapetable}{%
\clearpage
\noindent\begin{minipage}{\textwidth}
\begin{center}
\scriptsize
\setlength{\tabcolsep}{2.5pt}
\captionof{table}{Selected published results directly relevant to the present
theorem.  Weak typicality means \(S_n/\E S_n\to1\) in probability, whereas
strong typicality means \(S_n-\E S_n\to0\) in probability.  Results are
reported with the centers and assumptions used in the cited sources;
conditional statements and technical qualifications are marked explicitly.}
\label{tab:prior-landscape}
\begin{tabular}{@{}L{0.13\textwidth}L{0.20\textwidth}L{0.11\textwidth}L{0.21\textwidth}L{0.27\textwidth}@{}}
\toprule
Reference & Ensemble and squeezing & Entropy & Subsystem regime & Conclusion \\
\midrule
Serafini \emph{et al.} (2007)~\cite{SerafiniDahlstenGrossPlenio2007}, building on Ref.~\cite{SerafiniDahlstenPlenio2007}
& Microcanonical hard energy constraint and canonical Boltzmann measure on pure Gaussian states; squeezing variables are sampled
& von Neumann
& Fixed finite \(m\) as \(n\to\infty\), at fixed thermodynamic temperature \(T\)
& The reduced symplectic eigenvalues converge to \(1+T/2\); equivalently, the mean entropy converges to the corresponding \(m\) mode thermal entropy, and \(\Var S_m\to0\).  Ref.~\cite{SerafiniDahlstenGrossPlenio2007} measures entropy in bits.  Section~5.4 also gives numerical evidence at proportional size; see (e). \\
\addlinespace
Fukuda and Koenig (2019)~\cite{FukudaKoenig2019}
& Haar passive orbit of a deterministic, possibly unequal squeezing matrix \(Z_n\); \(\|Z_n\|_\infty\le Cn^\zeta\); \(\mu\le\lambda(n)\le C'\), \(\mu>1\)
& von Neumann
& \(k_n\le Kn^\kappa\); \(8\zeta+3\kappa<1\) and \(\zeta+3\kappa<1\)
& Theorem~3.10 proves additive concentration around the deterministic thermal value \(k_nG(\lambda(n))\), with a stretched exponential tail; see (a). \\
\addlinespace
Iosue \emph{et al.} (2023), Cor.~7~\cite{Iosue2023}
& Haar passive orbit; fixed equal \(s\ne0\)
& R\'{e}nyi entropy of order 2
& Fixed \(r=k/n\in(0,1)\) for weak; \(k_n=o(n)\) for strong
& Weak at fixed proportional ratio; strong in the stated sublinear regime. \\
\addlinespace
Iosue \emph{et al.} (2023), Cor.~8~\cite{Iosue2023}
& Haar passive orbit; fixed equal \(s\ne0\)
& von Neumann
& Every \(k_n=o(n)\)
& Corollary~8 states weak typicality.  No proportional subsystem conclusion; see (d). \\
\addlinespace
Iosue \emph{et al.} (2023), Rem.~11~\cite{Iosue2023}
& Haar passive orbit; deterministic, possibly unequal \(n\) dependent profile
& R\'{e}nyi entropy of order 2 and von Neumann
& R\'{e}nyi 2: \(k_n=o(n)\) or \(s_{\max,n}=o(1)\), where \(s_{\max,n}:=\max_j|s_{j,n}|\); von Neumann: \(k_n=o(n)\)
& Conditional claims of weak typicality assuming Conjecture~10; see (b) and (d). \\
\addlinespace
Youm \emph{et al.} (2025)~\cite{Youm2025}
& Haar passive orbit; fixed equal \(s\ne0\)
& Integer R\'{e}nyi \(\alpha\ge2\)
& Proportional \(k_n=\Theta(n)\), treated at fixed \(r=k/n\in(0,1)\)
& Weak typicality for every integer \(\alpha\ge2\); see (c). \\
\addlinespace
Youm \emph{et al.} (2025)~\cite{Youm2025}
& Haar passive orbit; fixed equal squeezing
& von Neumann
& Fixed \(r=k/n\in(0,1)\)
& Asymptotically exact mean Page curve; proportional typicality not established. \\
\addlinespace
This work
& Same fixed equal squeezing Haar ensemble, \(s\ne0\)
& von Neumann
& \(k_n/n\to r\in(0,1)\)
& Relative tail in Eq.~\eqref{eq:main-tail}; almost sure for every coupling with Haar marginals; \(\Var S_{1,n}=O_s(\log^2(en))\). \\
\bottomrule
\end{tabular}
\par\smallskip
\raggedright
\textit{Notes.}
(a) In the notation of Ref.~\cite{FukudaKoenig2019}, \(Z_n\) is the
deterministic squeezing matrix, \(\zeta\) bounds the growth of its operator
norm, \(\kappa\) bounds subsystem growth, \(\lambda(n)\) is the associated
mean energy parameter, the fixed constants \(\mu>1\) and \(C'\) give its
uniform lower and upper bounds, \(G\) is the one mode thermal entropy, and the source
denotes its inverse temperature factor by \(\beta(\mu)\).  Theorem~3.10 of
that reference assumes fixed exponents \(\zeta,\kappa\), centers the tail at
the deterministic thermal value
\(k_nG(\lambda(n))\), and gives
\(\exp[-c\varepsilon^4n^{1-8\zeta-3\kappa}/\beta(\mu)^4]\), above its
\(\varepsilon\) threshold.  This is the result summarized as strong typicality
in Refs.~\cite{Iosue2023,Youm2025}.  For \(\zeta=0\), \(\kappa<1/3\); the
theorem assumes a fixed exponent \(\kappa\), rather than every sequence
\(k_n=o(n^{1/3})\).
(b) Remark~11 of Ref.~\cite{Iosue2023} is conditional, not an unconditional
theorem.
(c) Section~IV.D and Appendix~B.4 of Ref.~\cite{Youm2025} establish
proportional weak typicality for every integer \(\alpha\ge2\).  Its Table~I
explicitly attributes the sublinear strong typicality column to
Ref.~\cite{Iosue2023}.
(d) Corollary~8 and the von Neumann part of Remark~11 are quoted as stated in
Ref.~\cite{Iosue2023}.  Their displayed argument gives
\(\Var S_{1,n}/n^2=O((k_n/n)^2)\).  Because Definition~5 normalizes relative
weak typicality by \((\E S_{1,n})^2\), this estimate alone does not establish
the stated relative convergence for an arbitrary sublinear sequence.  These
claims are not used in the present proof.
(e) Section~4 of Ref.~\cite{SerafiniDahlstenGrossPlenio2007} is a rigorous
fixed finite \(m\) theorem.  Its Section~5.4 separately studies
\(m/n\to r>0\): it states that both the mean and variance diverge and
reports numerical evidence that
\(\Var S_m/(\E S_m)^2\to0\).  That proportional size observation is numerical,
not part of the fixed \(m\) proof.
\end{center}
\end{minipage}
\clearpage}

Thus this paper closes the fixed equal squeezing, proportional size, von
Neumann weak typicality problem.  Strong typicality in fact
fails at every proportional ratio \(r\ne\tfrac12\), as proved from the Jacobi
linear statistics central limit theorem in Appendix~\ref{app:no-strong}.
Other recent work addresses different questions: Roy~\cite{Roy2025} studied a
different Gaussian ensemble, while Shou \emph{et al.}~\cite{Shou2026} studied
finite depth optical networks.

\subsection*{Main contributions}

The main result is the quantitative relative tail in
Eq.~\eqref{eq:main-tail}.  It yields weak typicality, a typical volume law,
almost sure convergence, the variance bound \(\Var S_{1,n}=O_s(\log^2 n)\),
logarithmic tightness, and concentration
about the published Page curve value at each fixed proportional ratio.  An accompanying Lean~4 development,
archived as version~1.0.0 on Zenodo~\cite{Zhao2026LeanVerification}, formalizes
the proof and its main corollaries.  Appendix~\ref{app:formal-scope} states the
exact verification scope.

\section{Model, definitions, and main result}
\label{sec:model}

\subsection{Gaussian boson sampling state}

We consider the ideal, lossless and noiseless premeasurement state, with pure
zero displacement squeezed vacuum inputs and a common squeezing axis, as in
the standard Gaussian boson sampling preparation of Ref.~\cite{Hamilton2017}.
The random state ensemble is exactly the equal squeezing specialization of the
fixed squeezing, passive Haar model of Iosue \emph{et al.}, whose setup is
given in Sec.~2 and App.~A.1 of Ref.~\cite{Iosue2023}.  Standard background and
the normalization conventions used below can be found in Sec.~II of
Ref.~\cite{Weedbrook2012} and Chs.~2--3 of Ref.~\cite{Serafini2017}.

For mode \(j\), let \(a_j\) and \(a_j^\dagger\) be the annihilation and creation
operators, with \([a_j,a_\ell^\dagger]=\delta_{j\ell}\), and define
\(q_j=a_j+a_j^\dagger\) and \(p_j=i(a_j^\dagger-a_j)\).  Operator hats are
suppressed throughout.  We use the quadrature ordering
\[
 R=(q_1,\ldots,q_n,p_1,\ldots,p_n)^{\mathsf T}
\]
with canonical commutation relations
\(
 [q_j,p_\ell]=2i\delta_{j\ell}
\)
and all other quadrature commutators zero.  This is the \(\hbar=2\)
normalization of Eqs.~(5)--(7) of Ref.~\cite{Weedbrook2012}.

For a state \(\rho\), write
\(\langle A\rangle_\rho=\tr(\rho A)\).  For phase space indices
\(a,b\in\{1,\ldots,2n\}\), let
\(\Delta R_a=R_a-\langle R_a\rangle_\rho\).  The \((a,b)\) entry of the
covariance matrix \(\sigma\) is
\[
 \sigma_{ab}=\frac12\left\langle
 \Delta R_a\Delta R_b+\Delta R_b\Delta R_a
 \right\rangle_\rho.
\]
Thus \(\sigma_{ab}\) is a symmetrized second moment, not a density matrix
element.  Gaussian states are determined by their first and second moments;
the covariance definition is Eq.~(15) of Ref.~\cite{Weedbrook2012} and agrees
with Sec.~2 of Ref.~\cite{Iosue2023}.  Displacements do not change the
subsystem entropy, so we set the first moments to zero; see Sec.~II.B.2 of
Ref.~\cite{Weedbrook2012}.  In this normalization the vacuum covariance is
\(I_{2n}\), the uncertainty relation is \(\sigma+i\Omega_n\succeq0\), where
\(\Omega_n=\left(\begin{smallmatrix}0&I_n\\-I_n&0\end{smallmatrix}\right)\)
in the chosen ordering, and physical symplectic eigenvalues satisfy
\(\nu_j\ge1\); see Eqs.~(17), (44), and (45) of
Ref.~\cite{Weedbrook2012}.

Each input mode is prepared in a single mode squeezed vacuum with the same
fixed nonzero real squeezing parameter \(s\).  Its magnitude \(|s|\) is the
squeezing strength, while its sign selects the squeezed quadrature.  A common
axis is chosen for convenience: fixed input squeezing phases can be absorbed
into the interferometer by right Haar invariance without changing the random
ensemble.  The input covariance is therefore
\begin{equation}
 \sigma_0=\diag(e^{2s}I_n,e^{-2s}I_n).
 \label{eq:input-covariance}
\end{equation}
This is the equal squeezing case of the covariance
\(Z\oplus Z^{-1}\) in App.~A.1 of Ref.~\cite{Iosue2023}, and it also follows
from Eqs.~(31) and (32) of Ref.~\cite{Weedbrook2012}; reversing the sign of
\(s\) only exchanges the squeezed quadrature.

A single input squeezed vacuum has mean photon number \(\sinh^2s\), so the
total input mean is \(n\sinh^2s\)~\cite{Weedbrook2012,Iosue2023}.  Passive
linear optics preserves total photon number.  The limit studied here therefore
keeps the mean photon density \(\sinh^2s\) fixed while the total mean photon
number grows linearly with \(n\); it is not a fixed total photon number limit.

For \(U=X+iY\in\U(n)\), the corresponding passive symplectic matrix is
\begin{equation}
 O_U=
 \begin{pmatrix}
 X&-Y\\Y&X
 \end{pmatrix},
 \qquad
 \sigma(U)=O_U\sigma_0O_U^{\mathsf T}.
 \label{eq:passive-map}
\end{equation}
The covariance transformation \(\sigma\mapsto O_U\sigma O_U^{\mathsf T}\)
is Eq.~(25) of Ref.~\cite{Weedbrook2012}.  It is also the passive
interferometer construction in Eqs.~(A1)--(A4) of Ref.~\cite{Iosue2023}.
Their realification \(\eta(U)\) equals our \(O_{\overline U}\); since complex
conjugation preserves normalized Haar measure on \(\U(n)\), the two sign
conventions define the same random ensemble.

The word ``typically'' is essential: for example, a real orthogonal \(U\)
commutes with the aligned input covariance through its realification, giving
\(\sigma(U)=\sigma_0\) and hence no modal entanglement.

The input is pure and passive Gaussian evolution preserves purity.  By output
permutation invariance of Haar measure, choosing the first \(k=k_n\) modes
entails no loss of generality.  We write \(\rho_{k,n}(U)\) for the partial
trace over the remaining \(n-k\) modes, exactly as in App.~A.1 of
Ref.~\cite{Iosue2023}; thus \(\rho_{k,n}(U)\) is the reduced density operator
of the first \(k\) output modes.  Its von Neumann entropy is
\begin{equation}
 S_{1,n}(U)=-\tr\!\left[\rho_{k,n}(U)\log\rho_{k,n}(U)\right].
 \label{eq:vn-definition}
\end{equation}
This is the standard trace definition; see Eq.~(26) of
Ref.~\cite{Weedbrook2012}.  Because the global output is pure,
\(S_{1,n}(U)\) is also the entanglement entropy across the chosen bipartition.
This interpretation relies on the ideal lossless preparation.  With loss, the
global optical state is generally mixed, and a reduced von Neumann entropy
contains both entanglement and ordinary mixedness rather than defining a mixed
state entanglement measure by itself.
Natural logarithms are used throughout, and \(\mathbb N_0=\{0,1,2,\ldots\}\).  We use the unnormalized Hilbert Schmidt norm
\(
 \|A\|_{\HS}^{2}=\tr(AA^*)=\sum_{a,b}|A_{ab}|^2
\), where \(A^*\) denotes the conjugate transpose.  For a real valued function
\(f\) on a metric space, \(\Lip(f)\) denotes its optimal global Lipschitz
constant.  Thus any displayed \(L\) satisfying
\(|f(x)-f(y)|\le Ld(x,y)\) gives \(\Lip(f)\le L\).  The notation
\(O_s(\cdot)\) means that the implicit constant may depend on the fixed
squeezing \(s\), but not on \(n\).  For deterministic \(a_n>0\),
\(X_n=O_{\Pp}(a_n)\) means that \((X_n/a_n)\) is bounded in probability,
equivalently that its laws form a tight family on \(\R\), in the terminology
of Sec.~2.2 of Ref.~\cite{vanDerVaart1998}.

\subsection{Weak and strong typicality}

Let \(m_n=\E S_{1,n}\), with expectation over normalized Haar measure on
\(\U(n)\).  Suppose that \(m_n>0\) for all sufficiently large \(n\).  The
entropy is \emph{weakly typical} if
\begin{equation}
 \frac{S_{1,n}}{m_n}\xrightarrow{\Pp}1.
 \label{eq:weak-typicality-definition}
\end{equation}
Following Definition~5 of Iosue \emph{et al.}~\cite{Iosue2023}, it is
\emph{strongly typical} if \(S_{1,n}-m_n\to0\) in probability.  Thus ``strong''
refers here to vanishing additive deviation.  Since \(m_n\) grows linearly for
a proportional bipartition, strong typicality implies weak typicality, but the
converse need not hold.  A \emph{coupling} of the sequence means random
variables \((U_n)\) on one probability space with the prescribed normalized
Haar marginal on \(\U(n)\) for every \(n\); no independence is included in
this definition.

\subsection{Main theorem}

\begin{theorem}[Von Neumann weak typicality]
\label{thm:main}
Fix \(s\in\R\setminus\{0\}\).  For every \(n\ge2\), let
\(k_n\in\{1,\ldots,n-1\}\) satisfy
\begin{equation}
 \frac{k_n}{n}\longrightarrow r\in(0,1).
 \label{eq:proportional-regime}
\end{equation}
For every \(n\ge2\), let \(U_n\) have normalized Haar distribution on \(\U(n)\).
Then
\begin{equation}
 \frac{S_{1,n}(U_n)}{\E S_{1,n}(U_n)}
 \xrightarrow{\Pp}1.
 \label{eq:main-convergence}
\end{equation}
More quantitatively, there exists \(c_{s,r}>0\) such that, for every
\(\varepsilon>0\), the following holds for all sufficiently large \(n\)
(with the threshold in \(n\) allowed to depend on \(\varepsilon\)):
\begin{equation}
 \Pp\!\left\{
 \left|\frac{S_{1,n}}{\E S_{1,n}}-1\right|\ge\varepsilon
 \right\}
 \le
 2\exp\!\left[-c_{s,r}
 \frac{\varepsilon^2n^2}{\log^2(en)}\right].
 \label{eq:main-tail}
\end{equation}
If \((U_n)\) is defined on a common probability space and each \(U_n\)
retains its normalized Haar marginal on \(\U(n)\), then the convergence in
Eq.~\eqref{eq:main-convergence} also holds almost surely.  No independence
across \(n\) is required.
\end{theorem}

For every allowed \(n\), Eq.~\eqref{eq:mean-lower-bound} below and \(s\ne0\)
give \(\E S_{1,n}>0\), so every normalized ratio in the theorem is defined.

The theorem applies to fixed nonzero equal squeezing, Haar random passive
interferometers, and proportional bipartitions.  It does not address unequal
squeezing, squeezing that vanishes with \(n\), finite depth optical circuits,
loss, or sampling hardness.  Appendix~\ref{app:no-strong} separately proves
that strong typicality fails when \(r\ne\tfrac12\).

\priorlandscapetable

The extensive mean bound proved below converts weak typicality into a direct
volume law statement.  Define
\(
 \gamma_{s,r}=\frac12\tanh^2(2s)r(1-r)>0
\).

\begin{corollary}[Typical extensive entanglement]
\label{cor:volume-law}
Under the assumptions of Theorem~\ref{thm:main}, let
\(0<v<\gamma_{s,r}\), and let \(h\) be the one mode entropy function in
Eq.~\eqref{eq:h-definition}.  Then
\[
 \Pp\!\left\{
 v n\le S_{1,n}(U_n)
 \le k_n\,h\!\left(\cosh(2s)\right)
 \right\}\longrightarrow1.
\]
For every coupling with the prescribed normalized Haar marginals,
\(S_{1,n}(U_n)\ge v n\) for all sufficiently large \(n\), almost surely.
The upper bound is deterministic: Eq.~\eqref{eq:symplectic-eigenvalues} gives
\(1\le\nu_j\le\cosh(2s)\), and monotonicity of \(h\) gives the displayed
estimate.  Hence
\(S_{1,n}=\Theta_{\Pp}(n)\) for a fixed proportional bipartition.
\end{corollary}

Purity sharpens the deterministic bound to
\[
 S_{1,n}(U_n)
 \le\min\{k_n,n-k_n\}\,h\!\left(\cosh(2s)\right)=O_s(n).
\]
{\color{black}
The first estimate follows from the \(k_n\) mode reduction.  Applying the same
estimate to the complementary \((n-k_n)\) mode reduction and using purity gives
the stated minimum.
\par}

The known average Page curve can now be promoted to a typical sample
statement.  Let \(\mathcal P_1(s,r)\) denote the limiting von Neumann
Page curve density calculated by Youm \emph{et al.} in Theorem~1 and the
discussion following their Eq.~(4)~\cite{Youm2025}.  For their fixed ratio
scaling \(k=rn\), along dimensions for which \(rn\) is integral, that published
result states
\begin{equation}
 \frac1n\E S_{1,n}\longrightarrow\mathcal P_1(s,r).
 \label{eq:known-page-curve}
\end{equation}

\begin{corollary}[Typical Page curve value]
\label{cor:page}
Specialize Theorem~\ref{thm:main} to the fixed ratio subsystem scaling of
Theorem~\ref{thm:published-page-mean}, for which the published result
Eq.~\eqref{eq:known-page-curve} holds.  Then
\begin{equation}
 \frac1n S_{1,n}(U_n)
 \xrightarrow{\Pp}\mathcal P_1(s,r),
 \label{eq:typical-page-curve}
\end{equation}
and the same convergence holds almost surely for every coupling with the
prescribed normalized Haar marginals.  In addition,
\(\mathcal P_1(s,r)\ge\gamma_{s,r}>0\).
\end{corollary}

\section{Exact reduction to a singular value statistic}
\label{sec:spectral}

Let
\begin{equation}
 \begin{aligned}
 P_{k,n}&=\begin{pmatrix}I_k&0\end{pmatrix}\in\C^{k\times n},\\
 B_U&=P_{k,n}UU^{\mathsf T}P_{k,n}^{\mathsf T}.
 \end{aligned}
 \label{eq:B-definition}
\end{equation}
The matrix \(Q_U=UU^{\mathsf T}\) has the COE distribution and \(B_U\) is
its leading \(k\times k\) principal block~\cite{Meckes2019}.  Let
\begin{equation}
 1\ge\tau_1(U)\ge\cdots\ge\tau_k(U)\ge0
 \label{eq:tau-definition}
\end{equation}
be the singular values of \(B_U\).  The upper bound follows because \(B_U\) is a compression of a unitary matrix.

Following the equal squeezing notation of Ref.~\cite{Iosue2023}, set
\begin{equation}
 A_s=\sinh(2s),\qquad B_s=\cosh(2s).
 \label{eq:cd-definition}
\end{equation}
Thus the squeezing matrices denoted by \(A\) and \(B\) in that reference are
\(A=A_sI_n\) and \(B=B_sI_n\), respectively.  The subscript \(s\) prevents
confusion between the coefficient \(B_s\) and the principal COE block \(B_U\).
Restricting Eq.~\eqref{eq:passive-map} to the selected modes gives
\begin{equation}
 \begin{aligned}
 \sigma_{k,n}(U)&=B_sI_{2k}+A_sM_U,\\
 M_U&=
 \begin{pmatrix}
 \Re B_U&\Im B_U\\
 \Im B_U&-\Re B_U
 \end{pmatrix}.
 \end{aligned}
 \label{eq:reduced-covariance}
\end{equation}
Equation~\eqref{eq:reduced-covariance} specializes the covariance block formula
in Appendix~A, Eq.~(A4), of Ref.~\cite{Iosue2023}.  In Appendix~B of
Ref.~\cite{Youm2025}, the identity before Eq.~(B6) gives the squared symplectic
matrix and Eq.~(B6) gives the R\'{e}nyi entropy.  Neither source states
Eq.~\eqref{eq:symplectic-eigenvalues}.
Introduce the symplectic form
\begin{equation}
 \Omega_k=
 \begin{pmatrix}0&I_k\\-I_k&0\end{pmatrix}.
 \label{eq:symplectic-form}
\end{equation}
Because \(B_U\) is complex symmetric,
\begin{equation}
 \Omega_kM_U=-M_U\Omega_k,
 \label{eq:anticommutation}
\end{equation}
while \(M_U^2\) is the realification of \(B_UB_U^*\).  Its eigenvalues are therefore
\begin{equation}
 \tau_1(U)^2,\tau_1(U)^2,\ldots,
 \tau_k(U)^2,\tau_k(U)^2.
 \label{eq:M-spectrum}
\end{equation}
Equation~\eqref{eq:anticommutation} implies
\begin{equation}
 \bigl(i\Omega_k\sigma_{k,n}(U)\bigr)^2
 =B_s^2I_{2k}-A_s^2M_U^2.
 \label{eq:symplectic-square}
\end{equation}
The first clause of External Theorem~\ref{thm:external-gaussian} supplies the
normal form input underlying this spectral interpretation.  Its classical
ingredients are the
unitary congruence factorization for complex symmetric matrices due to
Takagi~\cite{Takagi1924} and Williamson's symplectic normal form for positive
real matrices~\cite{Williamson1936}; a modern Gaussian state treatment, with
the vacuum symplectic eigenvalue normalized to one, is given in
Sections~3.2.3 and 3.2.4 of Ref.~\cite{Serafini2017}.
{\color{black}We use these normal form results specialized to the covariance
family displayed above.}
The spectrum of \(M_U\) is contained in \([-1,1]\), and
\(B_s-|A_s|=e^{-2|s|}>0\); hence \(\sigma_{k,n}\) is positive definite.  Therefore
\(i\Omega_k\sigma_{k,n}\) is similar to the Hermitian matrix
\(i\sigma_{k,n}^{1/2}\Omega_k\sigma_{k,n}^{1/2}\), so it is diagonalizable
with real spectrum.  Standard symplectic spectral pairing gives the pairs
\(\pm\nu_j\).  Mathematically, Eq.~\eqref{eq:symplectic-eigenvalues} follows
from Eqs.~\eqref{eq:anticommutation} and \eqref{eq:symplectic-square} and the
spectral definition of Williamson eigenvalues.  The positive symplectic
eigenvalues are
\begin{equation}
 \begin{aligned}
 \nu_j(U)
 &=\sqrt{B_s^2-A_s^2\tau_j(U)^2}\\
 &=\sqrt{1+A_s^2\bigl(1-\tau_j(U)^2\bigr)}.
 \end{aligned}
 \label{eq:symplectic-eigenvalues}
\end{equation}

For \(\nu\ge1\), define the one mode entropy function
\begin{equation}
 h(\nu)=
 \frac{\nu+1}{2}\log\frac{\nu+1}{2}
 -\frac{\nu-1}{2}\log\frac{\nu-1}{2},
 \label{eq:h-definition}
\end{equation}
with \(0\log0=0\).  The entropy of a bosonic Gaussian state is the sum of \(h\) over its symplectic eigenvalues~\cite{Weedbrook2012,Serafini2017}.
{\color{black}A direct derivation is as follows.}  For a Gaussian density operator
on \(k\) modes with finite second moments, the uncertainty relation gives
\(\nu_j\ge1\), and Williamson's decomposition is implemented by a Gaussian
unitary.  For \(\bar n\ge0\), define the one mode thermal density operator in
the Fock basis by
\[
 \rho_{\mathrm{th}}(\bar n)
 :=\sum_{m=0}^{\infty}
 \frac{1}{\bar n+1}
 \left(\frac{\bar n}{\bar n+1}\right)^m
 |m\rangle\!\langle m|.
\]
It has mean photon number \(\bar n\).  Displacement and the Williamson
Gaussian unitary do not change entropy, so the original state is isospectral
to
\[
 \bigotimes_{j=1}^{k}\rho_{\mathrm{th}}(\bar n_j),
 \qquad \bar n_j=\frac{\nu_j-1}{2}.
\]
Thus its one mode eigenvalues are
\[
 p_{j,m}=\frac{1}{\bar n_j+1}
 \left(\frac{\bar n_j}{\bar n_j+1}\right)^m,
 \qquad m\in\mathbb N_0.
\]
They are nonnegative and sum to one.  For \(\bar n_j>0\), summing the
geometric series and its first moment shows that their entropy is
\[
 -\sum_{m=0}^{\infty}p_{j,m}\log p_{j,m}=h(\nu_j).
\]
Indeed, both sides equal
\((\bar n_j+1)\log(\bar n_j+1)-\bar n_j\log\bar n_j\).
At \(\bar n_j=0\), the spectrum is \((1,0,0,\ldots)\) and both sides vanish under the convention \(0\log0=0\).  Finite tensor product additivity and unitary invariance then give \(S(\rho)=\sum_j h(\nu_j)\).  The one mode geometric spectrum and its entropy appear in Eqs.~(16) through (18) of Ref.~\cite{HolevoSohmaHirota1999}, while its Eqs.~(41) through (44) give the general Gaussian product normal form and entropy.  In the vacuum eigenvalue one convention used here, Eqs.~(26) and (27) of Ref.~\cite{Weedbrook2012} give the one mode thermal spectrum, Eqs.~(44) and (45) give the Williamson form and \(\nu_j\ge1\), and Eqs.~(46) through (49) give the entropy and Gaussian unitary thermal decomposition.

Define
\begin{equation}
 \phi_s(t)=h\!\left(\sqrt{1+A_s^2(1-t^2)}\right),
 \qquad 0\le t\le1.
 \label{eq:phi-definition}
\end{equation}
We have proved the following exact identity at finite \(n\):
\begin{equation}
 \boxed{
 S_{1,n}(U)=\sum_{j=1}^{k}\phi_s\bigl(\tau_j(U)\bigr).}
 \label{eq:entropy-spectral-statistic}
\end{equation}

This reduction isolates the endpoint obstruction to a uniform scalar
Lipschitz bound in the behavior of \(\phi_s\) at \(t=1\), corresponding to a
a Williamson symplectic eigenvalue approaching its pure value \(\nu=1\).
The proof does not differentiate ordered
singular values at degeneracies; it uses Mirsky's inequality below.

\section{Regularizing the endpoint singularity}
\label{sec:regularization}

Let
\begin{equation}
 \nu(t)=\sqrt{1+A_s^2(1-t^2)}.
 \label{eq:nu-of-t}
\end{equation}
For \(0\le t<1\), direct differentiation yields
\begin{equation}
 h'(\nu)=\frac12\log\frac{\nu+1}{\nu-1},
 \qquad
 \nu'(t)=-\frac{A_s^2t}{\nu(t)}.
 \label{eq:scalar-derivatives}
\end{equation}
Moreover,
\begin{equation}
 \nu(t)-1=\frac{A_s^2(1-t^2)}{\nu(t)+1},
 \label{eq:nu-minus-one}
\end{equation}
so that, using \(\nu(t)\le B_s\),
\begin{equation}
 \frac{\nu(t)+1}{\nu(t)-1}
 =\frac{(\nu(t)+1)^2}{A_s^2(1-t^2)}
 \le \frac{(B_s+1)^2}{A_s^2(1-t)}.
 \label{eq:log-ratio-bound}
\end{equation}
Since \(s\ne0\), the constant on the right is finite.  Define the continuous
endpoint modulus
\[
 g(x)=
 \begin{cases}
  x\log(e/x),&x>0,\\
  0,&x=0.
 \end{cases}
\]
Equations~\eqref{eq:scalar-derivatives} and \eqref{eq:log-ratio-bound} prove
the following endpoint estimate.

\begin{lemma}[Logarithmic endpoint]
\label{lem:endpoint}
For fixed \(s\ne0\), there is \(C_s>0\) such that
\begin{align}
 |\phi_s'(t)|
 &\le C_s\log\frac{e}{1-t},
 &&0\le t<1,
 \label{eq:phi-derivative-bound}\\
 0\le\phi_s(t)
 &\le C_s g(1-t),
 &&0\le t\le1.
 \label{eq:phi-value-bound}
\end{align}
\end{lemma}

After increasing \(C_s\) if necessary, the second inequality follows by
integrating the first and using \(\phi_s(1)=0\); the same enlarged constant is
used in both displays.  For \(0<\eta<1/2\), define
\begin{equation}
 \phi_{s,\eta}(t)=\phi_s\bigl(\min\{t,1-\eta\}\bigr),
 \qquad 0\le t\le1,
 \label{eq:regularized-phi}
\end{equation}
and the regularized entropy statistic
\begin{equation}
 F_{n,\eta}(U)
 =\sum_{j=1}^{k}\phi_{s,\eta}\bigl(\tau_j(U)\bigr).
 \label{eq:regularized-statistic}
\end{equation}
Because \(g(\eta)=\eta\log(e/\eta)\) for \(\eta>0\),
Lemma~\ref{lem:endpoint} gives
\begin{equation}
 \Lip(\phi_{s,\eta})
 \le C_s\log\frac e\eta
 \label{eq:scalar-lipschitz}
\end{equation}
and the deterministic approximation
\begin{equation}
 \sup_{U\in\U(n)}
 |S_{1,n}(U)-F_{n,\eta}(U)|
 \le C_sk\eta\log\frac e\eta.
 \label{eq:uniform-approximation}
\end{equation}

It remains to control the sensitivity to \(U\).  For \(U,V\in\U(n)\), compression and unitary invariance of the Hilbert Schmidt norm give
\begin{align}
 \|B_U-B_V\|_{\HS}
 &\le\|UU^{\mathsf T}-VV^{\mathsf T}\|_{\HS}\nonumber\\
 &\le2\|U-V\|_{\HS}.
 \label{eq:B-lipschitz}
\end{align}
External Theorem~\ref{thm:external-mirsky}, the Frobenius specialization of
Mirsky's singular value variation theorem~\cite{Mirsky1960} (see also the treatment of symmetric gauges and
unitarily invariant norms in Chapter~IV of Ref.~\cite{Bhatia1997}) states
that, for arbitrary complex matrices of the same size and with both complete
singular value lists in the same nonincreasing order,
\begin{equation}
 \sum_{j=1}^{k}
 |\tau_j(U)-\tau_j(V)|^2
 \le\|B_U-B_V\|_{\HS}^2.
 \label{eq:mirsky}
\end{equation}
Combining Eqs.~\eqref{eq:scalar-lipschitz} and \eqref{eq:mirsky} with the Cauchy Schwarz inequality yields
\begin{equation}
 |F_{n,\eta}(U)-F_{n,\eta}(V)|
 \le L_{n,\eta}\|U-V\|_{\HS},
 \label{eq:statistic-lipschitz}
\end{equation}
where
\begin{equation}
 L_{n,\eta}
 \le2C_s\sqrt{k}\log\frac e\eta.
 \label{eq:L-bound}
\end{equation}

\section{An extensive lower bound for the mean}
\label{sec:mean}

The concentration scale must be compared with the mean entropy.  We derive the needed linear lower bound without using the explicit Page curve formula.

Define the R\'{e}nyi entropy of order 2 by
\[
 S_{2,n}(U):=-\log\tr\!\left[\rho_{k,n}(U)^2\right].
\]
Monotonicity of R\'{e}nyi entropies in their order gives
\(S_{1,n}(U)\ge S_{2,n}(U)\).  From
Eq.~\eqref{eq:symplectic-eigenvalues},
\begin{equation}
 \begin{aligned}
 S_{1,n}(U)&\ge S_{2,n}(U),\\
 S_{2,n}(U)&=\frac12\sum_{j=1}^{k}
 \log\!\left[1+A_s^2(1-\tau_j(U)^2)\right].
 \end{aligned}
 \label{eq:S1-ge-S2}
\end{equation}
For \(0\le x\le1\),
\begin{equation}
 \log(1+A_s^2x)\ge\frac{A_s^2}{1+A_s^2}x.
 \label{eq:log-chord}
\end{equation}
Indeed, the derivative of the left hand side is
\(A_s^2/(1+A_s^2x)\ge A_s^2/(1+A_s^2)\) on this interval, and both sides vanish at
\(x=0\).
Consequently,
\begin{equation}
 S_{1,n}(U)
 \ge\frac12\tanh^2(2s)
 \sum_{j=1}^{k}\bigl(1-\tau_j(U)^2\bigr).
 \label{eq:entropy-energy-lower}
\end{equation}

The elementary COE second moment is
\begin{equation}
 \E |(Q_U)_{ab}|^2
 =\frac{1+\delta_{ab}}{n+1},
 \qquad Q_U=UU^{\mathsf T}.
 \label{eq:coe-second-moment}
\end{equation}
{\color{black}A direct derivation is as follows.}  For \(n\ge2\), permutation congruences reduce the second
moments to two numbers
\(a=\E|Q_{11}|^2\) and \(b=\E|Q_{12}|^2\), where we write
\(Q=Q_U\) for this calculation.  Diagonal phase congruences make
the mixed expectations among \(Q_{11},Q_{12},Q_{22}\) vanish.  Congruence by
the real $45^\circ$ rotation in the first two coordinates therefore gives
\[
 a=\frac14\E|Q_{11}+2Q_{12}+Q_{22}|^2
   =\frac14(2a+4b)=\frac a2+b,
\]
so $a=2b$.  Since every row of $Q$ is a unit vector,
$a+(n-1)b=1$, whence $b=1/(n+1)$ and $a=2/(n+1)$.  The case $n=1$
is immediate.  This proves Eq.~\eqref{eq:coe-second-moment} without a
Weingarten calculation or an irreducibility assumption.

Summing Eq.~\eqref{eq:coe-second-moment} over the leading block gives
\begin{equation}
 \E\sum_{j=1}^{k}\tau_j(U)^2
 =\E\|B_U\|_{\HS}^2
 =\frac{k(k+1)}{n+1}.
 \label{eq:expected-tau-square}
\end{equation}
Taking expectations in Eq.~\eqref{eq:entropy-energy-lower}, we obtain
\begin{equation}
 \boxed{
 \E S_{1,n}
 \ge\frac12\tanh^2(2s)
 \frac{k(n-k)}{n+1}.}
 \label{eq:mean-lower-bound}
\end{equation}
If \(k_n/n\to r\in(0,1)\), then there is \(b_{s,r}>0\) such that
\begin{equation}
 \E S_{1,n}\ge b_{s,r}n
 \label{eq:extensive-mean}
\end{equation}
for all sufficiently large \(n\).

\section{Concentration and proof of the main theorem}
\label{sec:concentration}

We apply External Theorem~\ref{thm:external-haar}, the concentration property
of normalized Haar measure on the full unitary group \(\U(n)\).  {\color{black}Its precise
formulation and references appear in Appendix~\ref{app:lean-boundary}.}
With the unnormalized chordal Hilbert Schmidt distance defined above, for
\(n\ge1\), normalized Haar \(U\in\U(n)\), \(L>0\), and any real function
\(G\) satisfying, for all \(U_1,U_2\in\U(n)\),
\[
 |G(U_1)-G(U_2)|\le L\|U_1-U_2\|_{\HS},
\]
the function is continuous, bounded, measurable, and integrable because \(\U(n)\) is compact, and for every \(t>0\) it satisfies the one sided, mean centered bound
\begin{equation}
 \Pp\{G-\E G\ge t\}
 \le\exp\!\left(-\frac{nt^2}{12L^2}\right),
 \qquad t>0.
 \label{eq:haar-concentration}
\end{equation}
Applying the same bound to \(-G\) and taking a union gives
\[
 \Pp\{|G-\E G|\ge t\}
 \le2\exp[-nt^2/(12L^2)].
\]

For \(n\ge2\), choose
\begin{equation}
 \eta_n=n^{-2}.
 \label{eq:eta-choice}
\end{equation}
Equations~\eqref{eq:uniform-approximation} and \eqref{eq:L-bound}, together with \(k\le n\), imply
\begin{align}
 R_n&:=\sup_U|S_{1,n}(U)-F_{n,\eta_n}(U)|\notag\\
 &\le C_s\frac{\log(en^2)}{n}=o(1),
 \label{eq:R-bound}\\
 L_{n,\eta_n}&\le2C_s\sqrt n\log(en^2).
 \label{eq:L-final}
\end{align}
Writing \(m_n=\E S_{1,n}\), the uniform comparison gives
\begin{equation}
 |S_{1,n}-m_n|
 \le
 |F_{n,\eta_n}-\E F_{n,\eta_n}|+2R_n.
 \label{eq:centered-comparison}
\end{equation}

Fix \(\varepsilon>0\).  By Eq.~\eqref{eq:extensive-mean} and \(R_n=o(1)\), for all sufficiently large \(n\),
\begin{equation}
 2R_n\le\frac{\varepsilon m_n}{2}.
 \label{eq:R-negligible}
\end{equation}
Using Eqs.~\eqref{eq:haar-concentration}, \eqref{eq:L-final}, and \eqref{eq:centered-comparison},
\begin{align}
 &\Pp\{|S_{1,n}-m_n|\ge\varepsilon m_n\}\nonumber\\
 &\quad\le
 \Pp\!\left\{
 |F_{n,\eta_n}-\E F_{n,\eta_n}|
 \ge\frac{\varepsilon m_n}{2}
 \right\}\nonumber\\
 &\quad\le
 2\exp\!\left[-c_{s,r}
 \frac{\varepsilon^2n^2}{\log^2(en)}\right].
 \label{eq:tail-conclusion}
\end{align}
This is Eq.~\eqref{eq:main-tail} and proves convergence in probability.  The
right hand side is summable in \(n\), so the first Borel Cantelli lemma proves
almost sure convergence for every coupling whose marginal at size \(n\) is
normalized Haar on \(\U(n)\).  No independence across dimensions is required.

The same argument yields a variance bound and a refinement of weak typicality.

\begin{corollary}[Variance and deviations below every polynomial scale]
\label{cor:variance}
Under the assumptions of Theorem~\ref{thm:main},
\begin{equation}
 \Var S_{1,n}=O_s\!\left(\log^2(en)\right).
 \label{eq:variance-bound}
\end{equation}
Moreover, for every fixed \(a>0\),
\begin{equation}
 \frac{S_{1,n}-\E S_{1,n}}{n^a}
 \xrightarrow{\Pp}0,
 \label{eq:subpolynomial-deviation}
\end{equation}
and the same convergence holds almost surely for every coupling with the
prescribed normalized Haar marginals.
\end{corollary}

\begin{proof}
Integrating the tail in Eq.~\eqref{eq:haar-concentration} gives
\(
 \Var F_{n,\eta_n}\le C L_{n,\eta_n}^2/n
 =O_s(\log^2(en)).
\)
Writing \(S_{1,n}=F_{n,\eta_n}+(S_{1,n}-F_{n,\eta_n})\), the uniform
error gives the explicit transfer
\[
 \Var S_{1,n}
 \le 2\Var F_{n,\eta_n}+2R_n^2
 =O_s(\log^2(en)).
\]
For Eq.~\eqref{eq:subpolynomial-deviation}, repeat the concentration argument with threshold \(t=n^a\); the resulting upper bound is of order
\(
 \exp[-c_{s,a}n^{2a}/\log^2(en)]
\), which is summable.
\end{proof}

The same estimate has a particularly transparent tightness formulation.

\begin{corollary}[Logarithmic tightness]
\label{cor:log-tightness}
Under the assumptions of Theorem~\ref{thm:main}, for every \(\delta>0\) there
exists \(C_{s,\delta}>0\) such that, for all sufficiently large \(n\),
\[
 \Pp\!\left\{
 |S_{1,n}-\E S_{1,n}|\ge C_{s,\delta}\log(en)
 \right\}\le\delta.
\]
Equivalently,
\(S_{1,n}-\E S_{1,n}=O_{\Pp}(\log(en))\).  In particular,
\[
 \frac{S_{1,n}-\E S_{1,n}}{\E S_{1,n}}
 =O_{\Pp}\!\left(\frac{\log n}{n}\right).
\]
\end{corollary}

\section{Discussion}
\label{sec:discussion}

\subsection{What the theorem adds to the Page curve}

An ensemble average alone is compatible with sample fluctuations of order \(n\)
or several distinct macroscopic entropy populations.  Theorem~\ref{thm:main}
rules out both possibilities.  Combining it with the published mean theorem,
Theorem~\ref{thm:published-page-mean}, shows that for each fixed ratio \(r\)
and the prescribed fixed ratio subsystem sequence, the entropy density of a
Haar distributed interferometer converges in probability to the Page curve
value \(\mathcal P_1(s,r)\).  This is pointwise in \(r\); it does not assert
simultaneous concentration over all ratios or all subsets of output modes.

For each fixed \(n\), if the two sets of output modes are assigned to two
parties, \(S_{1,n}\) is the standard entropy of entanglement of that bipartite
pure state.  The theorem concerns the separate limit in which the number of
modes grows.  We make no additional claim here about many copy entanglement
manipulation in infinite dimensional local Hilbert spaces.  Corollary~\ref{cor:volume-law}
gives a direct typical extensive entanglement statement without importing the
explicit Page formula.  Integer R\'{e}nyi entropies probe different portions of the
entanglement spectrum and need not determine the von Neumann entropy without
additional control.

\subsection{Why the logarithmic cutoff is sufficient}

The proof illustrates a potentially reusable cutoff strategy.  The scalar entropy profile is continuous, but its derivative has a logarithmic endpoint singularity.  Truncating within \(\eta\) of the singular value endpoint \(t=1\), equivalently the pure symplectic value \(\nu=1\), costs \(O(k\eta\log(1/\eta))\), whereas the Lipschitz constant grows only as \(O(\sqrt{k}\log(1/\eta))\).  With \(\eta=n^{-2}\), the deterministic cost vanishes while concentration remains strong enough to control extensive deviations.  A direct exact variance calculation is therefore unnecessary.

The method may apply to other spectral observables whose derivative has an integrable endpoint singularity.  It is less likely to work unchanged for observables with power law singularities, where the cutoff error and Lipschitz cost can compete at the extensive scale.

\subsection{Limitations and open directions}

First, Appendix~\ref{app:no-strong} shows that strong typicality fails for
every \(r\ne\tfrac12\).  At the balanced ratio \(r=\tfrac12\), the limiting
Jacobi support reaches the logarithmically singular entropy endpoint, so the
standard \(C^1\) Jacobi central limit theorem used there does not apply
directly.  The present concentration argument still proves weak typicality at
that ratio.  Determining its centered limiting law, as well as sharpening the
general \(O(\log^2 n)\) variance bound, remains open.

Second, equal squeezing makes the covariance reduction depend on one
principal block \(B_U\) of a COE matrix.  With unequal squeezing, the relevant
matrix couples the interferometer to a nonconstant diagonal squeezing profile,
and the reduction to a singular value statistic of one principal COE block is
lost.  Extending weak typicality to experimentally relevant unequal squeezing
is a natural next problem.

Third, Haar measure provides the maximally randomized passive unitary benchmark
against which finite depth architectures can be compared.  For one dimensional
random brickwall networks, Shou \emph{et al.} prove at most diffusive growth of
R\'{e}nyi entropy of order 2.  For more general circuit geometries they give
separate depth bounds for attaining large average subsystem entanglement and
for approaching Haar measure in Wasserstein distance~\cite{Shou2026}.  The
corresponding proportional subsystem von Neumann concentration problem remains
open.

Finally, the result is a statement about entanglement of the premeasurement Gaussian state.  It does not, by itself, establish sampling hardness or an anticoncentration property of output probabilities.  It supplies one rigorous component of the broader program of relating entanglement, circuit depth, and quantum sampling complexity.

\section{Conclusion}

We proved weak typicality of the von Neumann entanglement entropy for the ideal
pure premeasurement state in the fixed nonzero equal squeezing Haar ensemble
at proportional subsystem size.  The
proof closes that Haar random \(k=\Theta(n)\) cell explicitly left open in
Table~I of Ref.~\cite{Youm2025}, without evaluating the squared
hypergeometric series in a direct variance expansion.  Its ingredients are an
exact principal COE block singular value representation, a vanishing endpoint
cutoff, a self contained extensive mean lower bound, and Haar concentration.
The result yields a typical volume law, logarithmic tightness, and
\(\Var S_{1,n}=O_s(\log^2n)\).  Together with the published mean Page curve of
Youm \emph{et al.}~\cite{Youm2025}, it promotes the Page curve value at each
fixed proportional bipartition to a typical sample statement.

\appendix

\section{\textcolor{black}{External mathematical inputs used in the formal verification}}
\label{app:lean-boundary}

{\color{black}
The formalization represents three groups of cited results as explicit theorem
assumptions: full unitary Haar concentration, Mirsky's singular value
inequality, and the Gaussian covariance and normal form results used in the
spectral reduction.  The Gaussian assumptions are displayed below in two
clauses: the covariance and symplectic spectrum statements form clause~(i),
and the thermal normal form is clause~(ii).  The independently published mean
limit in Eq.~\eqref{eq:known-page-curve}, used only for the Page curve
corollary, is stated separately as
Theorem~\ref{thm:published-page-mean}.
\par}

\setcounter{theorem}{0}
\renewcommand{\thetheorem}{A.\arabic{theorem}}
\renewcommand{\theHtheorem}{appendix.\arabic{theorem}}

\begin{theorem}[Full unitary Haar concentration]
\label{thm:external-haar}
Let \(n\ge1\), let \(U\) have normalized Haar distribution on \(\U(n)\),
and let \(G:\U(n)\to\R\) be integrable.  Suppose that \(L>0\) and
\begin{equation}
 |G(U_1)-G(U_2)|\le L\|U_1-U_2\|_{\HS}
 \label{eq:external-haar-lipschitz}
\end{equation}
for all \(U_1,U_2\in\U(n)\), where
\(
 \|A\|_{\HS}^2=\tr(AA^*)=\sum_{a,b}|A_{ab}|^2
\)
is the unnormalized chordal Hilbert Schmidt norm.  Then, for every \(t>0\),
\begin{equation}
 \Pp\!\left\{G(U)\ge \E G(U)+t\right\}
 \le \exp\!\left(-\frac{nt^2}{12L^2}\right).
 \label{eq:external-haar-tail}
\end{equation}
\end{theorem}

The exact one factor statement is Corollary~17 of Meckes and
Meckes~\cite{MeckesMeckes2013Unitary}; their Theorem~15 supplies the
full \(\U(n)\) log Sobolev constant \(6/n\).  The monograph derivation uses
Theorem~5.16 of Ref.~\cite{Meckes2019} and the one sided part of the Herbst
argument in its Theorem~5.5.  The displayed statement of that book theorem is
two sided and consequently has prefactor \(2\); Eq.~\eqref{eq:external-haar-tail}
is the one sided form.  {\color{black}Equivalent formulations with the same one sided
convention and denominator \(12L^2\) also appear in Supplemental Appendix~A,
Theorem~2, of Ref.~\cite{SinghKorbiczCerf2023} and Appendix~A, Theorem~6, of
Ref.~\cite{NakataWakakuwaKoashi2023}.}  {\color{black}Theorem~\ref{thm:external-haar}
uses the unnormalized chordal Hilbert Schmidt metric; normalized Frobenius or
geodesic conventions would change the constants.}

\begin{theorem}[Mirsky's complete list Frobenius inequality]
\label{thm:external-mirsky}
Let \(n\in\mathbb N_0\) and \(A,B\in\C^{n\times n}\).  Write their complete
singular value lists, including zeros and in the same nonincreasing order, as
\begin{align*}
 s_1(A)&\ge\cdots\ge s_n(A)\ge0,\\
 s_1(B)&\ge\cdots\ge s_n(B)\ge0.
\end{align*}
Then
\begin{equation}
 \begin{aligned}
 \sum_{j=1}^{n}\bigl(s_j(A)-s_j(B)\bigr)^2
 &\le \|A-B\|_{\HS}^2,\\
 \|A-B\|_{\HS}^2
 &=\sum_{a,b=1}^{n}|A_{ab}-B_{ab}|^2.
 \end{aligned}
 \label{eq:external-mirsky}
\end{equation}
No invertibility, rank, probability, or simple spectrum hypothesis is made.
\end{theorem}

This is the Euclidean gauge/Frobenius specialization of Mirsky's
singular value variation theorem~\cite{Mirsky1960}; squaring its nonnegative
norm inequality gives Eq.~\eqref{eq:external-mirsky}.  Chapter~IV of
Ref.~\cite{Bhatia1997} gives the standard textbook framework of symmetric
gauge functions and unitarily invariant norms used in this specialization.

\begin{theorem}[Gaussian normal forms used in the reduction]
\label{thm:external-gaussian}
The Gaussian input consists of the following two clauses, in the convention
that the vacuum symplectic eigenvalue equals one.

\emph{(i) Equal squeezing symplectic spectrum.}
Let \(k\in\mathbb N_0\), \(s\in\R\), and let
\(B\in\C^{k\times k}\) be a complex symmetric contraction.  Write
\(B^{\mathsf T}=B\) and let
\(1\ge\tau_1\ge\cdots\ge\tau_k\ge0\) be its singular values.  Set
\[
 A_s=\sinh(2s),\qquad B_s=\cosh(2s),
\]
and
\[
 M(B)=\left(\begin{smallmatrix}
 \Re B&\Im B\\ \Im B&-\Re B
 \end{smallmatrix}\right),
\]
and define
\(
 \sigma(B;s)=B_sI_{2k}+A_sM(B).
\)
This matrix is a valid Gaussian covariance matrix in the stated convention.
Every Gaussian state with this covariance has, in the order induced by the
singular values, nondecreasing Williamson symplectic eigenvalues
\(1\le\nu_1\le\cdots\le\nu_k\), and for \(j=1,\ldots,k\),
\begin{equation}
 \nu_j(B;s)
 =\sqrt{1+\sinh^2(2s)\bigl(1-\tau_j^2\bigr)}.
 \label{eq:external-gaussian-symplectic}
\end{equation}

\emph{(ii) Thermal spectrum.}
Let \(\rho\) be a Gaussian density operator on finitely many modes, with \(k\)
modes and finite
second moments and Williamson symplectic eigenvalues
\(\nu_1,\ldots,\nu_k\).  Then \(\nu_j\ge1\), and the complete density operator
spectrum, indexed by
\(\boldsymbol n=(n_1,\ldots,n_k)\in\mathbb N_0^k\), is
\begin{equation}
 \lambda_{\boldsymbol n}
 =\prod_{j=1}^{k}(1-q_j)q_j^{n_j},
 \qquad q_j=\frac{\nu_j-1}{\nu_j+1},
 \label{eq:external-gaussian-thermal}
\end{equation}
with the convention \(0^0=1\).
\end{theorem}

The classical ingredients for part~(i) are Takagi's unitary congruence
factorization~\cite{Takagi1924} and Williamson's symplectic normal
form~\cite{Williamson1936}; Sections~3.2.3 and 3.2.4 of
Ref.~\cite{Serafini2017} give a modern Gaussian state treatment.  For the
specific ensemble, specializing and translating the conventions of
Appendix~A, Eq.~(A4), of Ref.~\cite{Iosue2023} gives the reduced covariance.
The displayed identity immediately preceding Eq.~(B6) in Appendix~B of
Ref.~\cite{Youm2025} gives the corresponding squared symplectic matrix;
Eq.~(B6) is the subsequent R\'{e}nyi entropy expression.  Equation
\eqref{eq:external-gaussian-symplectic} also uses the elementary identification
of the spectrum of the realification of \(BB^*\) with the squared singular
values of \(B\).  {\color{black}Together, these results yield the specialized
identity in Eq.~\eqref{eq:external-gaussian-symplectic}.}

For part~(ii), Eqs.~(16) through (18) and (41) through (44) of
Ref.~\cite{HolevoSohmaHirota1999} give the geometric one mode spectrum and the
general Gaussian product normal form.  In the same vacuum eigenvalue one
convention, Eqs.~(26) and (27) and Eqs.~(44) through (49) of Ref.~\cite{Weedbrook2012} give
the one mode thermal spectrum, Williamson form, uncertainty bound, entropy,
and Gaussian unitary thermal decomposition; Sections~3.2.3, 3.2.4, and 3.5 of
Ref.~\cite{Serafini2017} give a textbook account.

{\color{black}
Within the formalization, Eq.~\eqref{eq:external-gaussian-thermal} is used to
prove normalization of the geometric spectra, evaluate the entropy sums at
both thermal and pure endpoints, and establish finite product additivity.  This
gives
\par}
\begin{equation}
 -\sum_{\boldsymbol n}\lambda_{\boldsymbol n}
    \log\lambda_{\boldsymbol n}
 =\sum_{j=1}^{k}h(\nu_j).
 \label{eq:external-gaussian-entropy-derived}
\end{equation}
{\color{black}
Thus the Gaussian interface assumes the reduced covariance formula, its
Williamson symplectic spectrum, and the Gaussian thermal normal form, whereas
Eq.~\eqref{eq:external-gaussian-entropy-derived} is derived within the
formalization.
\par}

\begin{theorem}[Published mean von Neumann Page curve]
\label{thm:published-page-mean}
Let \(s\in\R\) and \(r\in[0,1]\) be fixed.  In the equal squeezing Haar
ensemble, take \(k=rn\) along dimensions for which \(rn\) is integral.  Then
Youm \emph{et al.}, Theorem~1 and the paragraph following their
Eq.~(4)~\cite{Youm2025}, prove Eq.~\eqref{eq:known-page-curve}, with
\[
 \mathcal P_1(s,r)
 =\sum_{i=1}^{\infty}\left[\frac{1}{2i}-D_i(s)\right]G_i(r),
\]
where
\[
 \begin{aligned}
 D_i(s)&=\frac{1}{3}\operatorname{sech}^{2}(2s)\tanh^{2i}(2s)\\
 &\quad\times{}_2F_1\!\left(\frac32,1+i;\frac52;
                         \operatorname{sech}^{2}(2s)\right),\\
 G_i(r)&=r-r^{i+1}C_i\,{}_2F_1(1-i,i;2+i;r),\\
 C_i&=\frac{1}{i+1}\binom{2i}{i}.
 \end{aligned}
\]
Their finite size formula also contains an \(n\) independent term
\(H_i(r)=4^{i-1}[r(1-r)]^i\) and an \(o(1)\) remainder; both vanish after
division by \(n\).  Thus Eq.~\eqref{eq:known-page-curve} is a cited theorem in
{\color{black}the present argument and supplies the mean limit used in
Corollary~\ref{cor:page}.}
\end{theorem}

{\color{black}
The formalization does not reproduce the hypergeometric calculation in this
theorem; it represents Eq.~\eqref{eq:known-page-curve} as an explicit theorem
assumption for the Page curve transfer.
\par}

\section{Formal verification scope and reproducibility}
\label{app:formal-scope}

{\color{black}
The accompanying Lean~4 development records a corresponding declaration and
verification status for each of the 57 labeled equations.  The results in
Theorems~\ref{thm:external-haar}--\ref{thm:external-gaussian} are represented as
explicit theorem assumptions, and the remaining theorem level relations in
the core proof are kernel checked within the spectral Gaussian state
interface.  The Page curve transfer separately assumes the published mean
limit in Theorem~\ref{thm:published-page-mean}.

The formalization transcribes the trace definition in
Eq.~\eqref{eq:vn-definition}, but it does not construct bosonic Fock space,
trace class density operators, partial traces, or operator logarithms.
Consequently, the operator trace formulation and the abstract spectral model
are not identified internally.  Theorem~\ref{thm:external-gaussian} supplies
the covariance and spectral bridge; from the thermal spectrum onward, the
entropy calculation and the finite matrix, probability, and asymptotic
arguments are formalized.
The non-typicality result in Appendix~\ref{app:no-strong} is outside the scope
of the formalization.
\par}

The package pins stable Lean~4.33.0 and the following Mathlib commit:
\begin{center}
 \footnotesize\texttt{641fbd329d4ffb62bef83c51f54088469056bd36}.
\end{center}
Build instructions and evidence are in \path{REPRODUCIBILITY.md} and
\path{BUILD_EVIDENCE.md}; the equation and axiom audits are
\path{scripts/check_equation_labels.sh} and
\path{Paper/EquationEndpointsAxiomAudit.lean}.  The package includes clean
build CI and omits toolchain specific \texttt{.olean} files.

\section{Failure of strong typicality away from the balanced cut}
\label{app:no-strong}

This appendix uses an additional random matrix input that is not needed for
the weak typicality theorem: the central limit theorem for smooth linear
statistics of the \(\beta\)-Jacobi ensemble.  It determines enough of the order
one fluctuation scale to rule out strong typicality, equivalently vanishing
additive deviation from the mean.

\setcounter{theorem}{0}
\renewcommand{\thetheorem}{C.\arabic{theorem}}
\renewcommand{\theHtheorem}{appendixC.\arabic{theorem}}

\begin{theorem}[Failure of strong typicality]
\label{thm:no-strong}
Fix \(s\in\R\setminus\{0\}\), and let \(k_n/n\to r\in(0,1)\) with
\(r\ne\tfrac12\).  If \(U_n\) has normalized Haar distribution on
\(\U(n)\), then
\[
 S_{1,n}(U_n)-\E S_{1,n}(U_n)
 \not\xrightarrow{\Pp}0.
\]
Thus the von Neumann entropy is not strongly typical at any fixed nonbalanced
proportional bipartition: its additive deviation from the mean does not vanish.
\end{theorem}

\begin{proof}
Put
\[
 \begin{aligned}
 r_\star&=\min\{r,1-r\}\in(0,\tfrac12),\\
 \ell_n&=\min\{k_n,n-k_n\}.
 \end{aligned}
\]
Because the global Gaussian state is pure, the two complementary reduced
states have the same nonzero spectrum.  We may therefore compute the entropy
on the smaller subsystem, for which \(\ell_n/n\to r_\star\).  Write
\(m=\ell_n\), set
\(Q_n=U_nU_n^{\mathsf T}\), and let \(B_n\) be the leading \(m\) by \(m\)
block of \(Q_n\).  If \(R_{j,n}=\tau_{j,n}^2\) are the squared singular values
of \(B_n\), Eq.~\eqref{eq:entropy-spectral-statistic} gives
\[
 \begin{aligned}
 S_{1,n}(U_n)&=\sum_{j=1}^{m} f_s(R_{j,n}),\\
 f_s(x)&=h\!\left(\sqrt{1+A_s^2(1-x)}\right),\\
 A_s&=\sinh(2s).
 \end{aligned}
\]

The matrix \(Q_n\) is COE distributed.  Partitioning it into reflection and
transmission blocks, unitarity identifies \(1-R_{j,n}\) with the transmission
eigenvalues.  Equation~(1.7) of Forrester~\cite{Forrester2006Jacobi}, with
Dyson index \(\beta=1\), therefore gives the exact joint density
\[
 \begin{aligned}
 p_{m,n}(R)
 &=\frac{1}{Z_{m,n}}
   \prod_{1\le i<j\le m}|R_i-R_j|\\
 &\quad\times
   \prod_{j=1}^{m}(1-R_j)^{(n-2m-1)/2}.
 \end{aligned}
\]
This density is supported on \(R\in(0,1)^m\).
In the notation of Eq.~(1) of Dumitriu and
Paquette~\cite{DumitriuPaquette2012}, this is the real \(\beta\)-Jacobi
ensemble of matrix size \(m\), with
\[
 n_1=m+1,
 \qquad n_2=n-m,
 \qquad \beta=1.
\]
Consequently \(n_1/m\to p=1\) and
\(n_2/m\to q=(1-r_\star)/r_\star>1\).  Proposition~1.1 of
Ref.~\cite{DumitriuPaquette2012} gives the limiting support
\[
 [\lambda_-,\lambda_+]=[0,b_{r_\star}],
 \qquad b_{r_\star}=4r_\star(1-r_\star)<1.
\]
The real Jacobi largest root theorem of
Johnstone~\cite{Johnstone2008Jacobi} further gives, in this proportional
regime,
\[
 \max_{1\le j\le m}R_{j,n}\xrightarrow{\Pp}b_{r_\star}.
\]

Choose \(\delta>0\) with \(b_{r_\star}+2\delta<1\).  Since the only loss of
regularity of \(f_s\) on \([0,1]\) occurs at \(x=1\), there is a function
\(F_s\in C^1([0,1])\) that agrees with \(f_s\) on
\([0,b_{r_\star}+\delta]\).  Define the smooth Jacobi statistic
\[
 Y_n=\sum_{j=1}^{m}F_s(R_{j,n}).
\]
The largest root convergence implies
\[
 \Pp\{Y_n=S_{1,n}(U_n)\}\longrightarrow1,
 \qquad
 Y_n-S_{1,n}(U_n)\xrightarrow{\Pp}0.
\]

Theorem~1.3 of Dumitriu and Paquette~\cite{DumitriuPaquette2012} now yields
\[
\begin{aligned}
Y_n-\E Y_n&\xrightarrow{\mathrm d}
  \mathcal N(0,\sigma_{s,r_\star}^2),\\
\sigma_{s,r_\star}^2
  &=2\sum_{\ell=1}^{\infty}\ell
    |\widehat F_s(\ell)|^2.
\end{aligned}
\]
Here the shifted Chebyshev coefficients on \([0,b_{r_\star}]\) are
\[
 \begin{aligned}
 \widehat F_s(\ell)
 &=\frac{1}{2\pi}\int_0^{b_{r_\star}}
   \frac{F_s(x)\,\Gamma_\ell(x)}
        {\sqrt{(b_{r_\star}-x)x}}\,dx,\\
 \Gamma_\ell(x)
 &=2T_\ell\!\left(\frac{2x-b_{r_\star}}{b_{r_\star}}\right).
 \end{aligned}
\]
Assumption~1.2 of that reference is printed as
\(n_1=pm\), \(n_2=qm\) with fixed \(p,q\).  The triangular array form used
here, including \(n_1=m+1\) and \(n_2/m\to q\), follows from the same proof:
the
polynomial fluctuation calculation is termwise continuous in the two ratios,
and the \(C^1\) approximation and Poincar\'{e} estimates in its Section~4 are
uniform when the ratios remain in a compact subset of
\(p,q\ge1\), \(p+q>2\).  Thus the triangular array above has the same
limiting covariance at \((p,q)=(1,(1-r_\star)/r_\star)\).

The limiting variance is strictly positive.  Indeed, \(s\ne0\) makes
\(f_s\), hence \(F_s\) on \([0,b_{r_\star}]\), strictly nonconstant.  If every
coefficient \(\widehat F_s(\ell)\) with \(\ell\ge1\) vanished, completeness of
the shifted Chebyshev system would make \(F_s\) constant almost everywhere
for the arcsine weight on that interval, and continuity would make it
constant everywhere, a contradiction.  Hence
\(\sigma_{s,r_\star}^2>0\).

It remains to compare the two centerings.  Suppose, to the contrary, that
\(X_n=S_{1,n}(U_n)\) satisfied \(X_n-\E X_n\to0\) in probability.  Set
\(d_n=\E X_n-\E Y_n\).  Since \(Y_n-X_n\to0\) in probability,
\[
 (Y_n-\E Y_n)-d_n
 =(Y_n-X_n)+(X_n-\E X_n)
 \xrightarrow{\Pp}0.
\]
The Gaussian convergence makes \(Y_n-\E Y_n\) tight, so the deterministic
sequence \(d_n\) must be bounded.  Along a subsequence on which \(d_n\to d\),
the last display would force \(Y_n-\E Y_n\) to converge in distribution to
the point mass at \(d\).  This contradicts its convergence to the
nondegenerate Gaussian \(\mathcal N(0,\sigma_{s,r_\star}^2)\).  Therefore
strong typicality fails.
\end{proof}

The exclusion of \(r=\tfrac12\) is a limitation of this particular argument,
not evidence for strong typicality at the balanced cut.  There
\(b_{r_\star}=1\), so the limiting Jacobi support reaches the endpoint where
\(f_s'(x)\) diverges logarithmically, and the cited \(C^1\) central limit
theorem cannot be applied directly.  Theorem~\ref{thm:no-strong} and its
{\color{black}Jacobi inputs are used only for this additional strong typicality
conclusion.}

\section*{Author contributions and use of artificial intelligence}

H.Z. conceived the project, developed and checked the mathematical argument,
prepared the formalization, and wrote and revised the manuscript.  OpenAI
ChatGPT and Codex tools were used under the author's direction for literature
search assistance, proof exploration, Lean~4 code drafting, algebraic
cross checking, and language editing.  The resulting Lean source was compiled
and checked by the Lean kernel.  The author reviewed all mathematical
statements, citations, source files, and final text and accepts full
responsibility for the work.

\section*{Data and code availability}

No experimental data were generated.  The Lean~4 source archive contains the
equation registry, permanent kernel audits, build instructions, and CI.  It
{\color{black}contains complete formal proofs for the formalized results and records the cited
external results as explicit theorem assumptions.}  It is released under
\texttt{GPL-3.0-only}; that license does not cover the
manuscript or third party files.  The versioned archive is on
Zenodo~\cite{Zhao2026LeanVerification}:
\url{https://doi.org/10.5281/zenodo.21969265}.

\bibliographystyle{unsrt}
\bibliography{references}

\end{document}